\documentclass[preprint]{elsarticle} %numberref for arXiv
\usepackage{url}
\usepackage{mathtools}
\usepackage{mathrsfs}
\usepackage{algorithm}
\usepackage{algpseudocode}
\usepackage{amsmath,amssymb,amsthm}

\newtheorem{thm}{Theorem}
\newtheorem{lem}[thm]{Lemma}
\newproof{pf}{Proof}

\begin{document}

\title{Symmetric solution of the Bellman optimality equation for repeated harmony game }

\author[inst1,inst2]{Hisato Komatsu}

\ead{komatsu@phys.kindai.ac.jp}

\affiliation[inst1]{ organization={Department of Physics, Kindai University},
postcode={577-8502},
addressline={Higashi-Osaka}, 
city={Osaka},
country={Japan},}

\affiliation[inst2]{ organization={Data Science and AI Innovation Research Promotion Center, Shiga University},
%addressline={},
postcode={522-8522},
addressline={Hikone}, 
city={Shiga},
country={Japan} }%Lines break automatically or can be forced with \\

\begin{abstract}
In social dilemma games, additional rewards or punishments have been studied as means of promoting cooperation. Therefore, it is important to investigate the ideal situation, in which such an additional payoff would change the game.
In this study, we investigated the symmetric solution of the Bellman optimality equation for a repeated harmony game. The calculations showed that three types of symmetric solutions exist. One of them corresponds to the trivial All-C strategy, and another to the Win-stay Lose-shift strategy of the prisoners dilemma game. The nontrivial behavior of the strategy corresponding to the last solution is also discussed in detail. In addition, we numerically investigated which strategy the agents actually learn by the reinforcement learning algorithm.
\end{abstract}

\maketitle

\section{Introduction}

The methods for promoting cooperation in social dilemma games are important in game theory. One such method is the introduction of additional rewards or punishments by participants \cite{Fehr2000, Boyd2003} or third parties \cite{Fehr2004,Kwaadsteniet2013}. Along with recent advances in machine learning, multi-agent reinforcement learning methods incorporating such additional rewards have been proposed \cite{Yang2020}. Moreover, the learning dynamics of agents in the presence of additional rewards have also been studied \cite{Zhao2024}.
Various studies have been conducted on reinforcement learning applied to repeated games \cite {Sandholm1996, Forester2018, Hu2023}. Usui and Ueda investigated the symmetric solution, i.e., the solution in which the two agents have the same strategy, of the Bellman optimality equation for the repeated prisoner's dilemma \cite{Usui2021}. They showed that the All-D strategy creates an equilibrium regardless of the conditions, and the Win-stay Lose-shift (WSLS) \cite{Nowak1993} and Grim strategies create it when specific conditions are fulfilled. The Bellman optimality equation was also employed by Meylahn and Janssen to investigate the best response to an opponent with a fixed strategy \cite{Meylahn2022, Meylahn2025}.

The harmony game is a $2 \times 2$ game in which cooperation yields a higher payoff for an agent regardless of the opponent's strategy. This game is often excluded from theoretical considerations because cooperation is the dominant strategy, at least in one-shot games. However, several previous studies have reported that, even in this game, cooperation can collapse under certain evolutionary \cite{Roca2006} or learning dynamics \cite{Couto2026}, or the spatial structure \cite{Matamalas2015}, indicating that the emergence and maintenance of cooperation is not always trivial in practice. Moreover, in repeated games, the folk theorem ensures the existence of strategies that achieve a wide range of average payoffs \cite{Fudenberg1986}.
Therefore, before considering situations in which additional payoffs are introduced to promote cooperation, we first examine the ideal case in which such additional payoffs completely transform the game into a harmony game. This study investigated the Bellman optimality equation of the harmony game with an approach similar to that in Ref.~\cite{Usui2021} and lists possible strategies corresponding to symmetric solutions. In addition, we performed numerical experiments using reinforcement learning (Q-learning \cite{Watkins,Sutton}) to examine which strategies the agents learn under several initial conditions.

The remainder of this paper introduces this model in Sec.~\ref{model} and then derives the symmetric solutions of the Bellman optimality equation in Sec.~\ref{Bellman}. We compare these solutions with the numerical simulations in Sec.~\ref{experiments}, and summarize the study in Sec.~\ref{Summary}. For reproducibility, details of the simulations in Sec.~\ref{experiments} are provided in \ref{simulation_details}. We also discuss the optimal strategy when the opponent's strategy is fixed, in \ref{best_response}.

\section{Model \label{model} }

We consider the repeated harmony game, in which two agents (A and B) choose an action, either $C$ (cooperation) or $D$ (defection), at each time step. The reward for each agent, $r _A$ or $r _B$, depends on their joint action, $\left( a, b \right)$, according to
\begin{equation}
\begin{aligned}
r _A \left( \left. C, C \right. \right) & = r _B \left( \left. C, C \right. \right) = R , \\
r _A \left( \left. C, D \right. \right) & = r _B \left( \left. D, C \right. \right) = S , \\
r _A \left( \left. D, C \right. \right) & = r _B \left( \left. C, D \right. \right) = T , \\
r _A \left( \left. D, D \right. \right) & = r _B \left( \left. D, D \right. \right) = P ,
\end{aligned}
\end{equation}
with $R > T > P$ and $R > S > P$. Strictly speaking, the definition of the harmony game varies slightly across references. Our definition is identical to that used in Ref.~\cite{Matamalas2015}, for example. Throughout this paper, we use the term ``reward'' instead of ``payoff'' to align with the terminology of reinforcement learning.
In the one-shot game, the unique Nash equilibrium is $\left( C, C \right)$, and $C$ is the dominant strategy of both agents. Moreover, unlike the prisoner's dilemma, the one-shot game lacks a social dilemma. However, in a repeated game, the folk theorem ensures the existence of Nash equilibrium strategies other than the trivial All-C strategy \cite{Fudenberg1986}. Therefore, theoretical and numerical investigations should be carefully conducted to avoid overlooking such nontrivial strategies.
For simplicity, we focus on memory-1 strategies, in which agents forget any actions taken two steps ago or earlier and choose their current action based solely on their previous action. In this case, the observation of the agent was the previous joint action.

\section{Bellman optimality equation and its symmetric solution \label{Bellman} }

By ignoring exploration noise such as $\epsilon$ in $\epsilon$-greedy method, the Bellman optimality equation is expressed as follows:
\begin{eqnarray}
Q_A ^{\ast} \left( \left( a_0, b_0 \right), a_1 \right) & = & r_A \left( \left. a_1, \tilde{b} _1 \right. \right) + \gamma \max _{a _2} Q_A ^{\ast} \left( \left( a_1, \tilde{b}_1 \right) , a _2 \right) , \label{Bellman1A} \\
Q_B ^{\ast} \left( \left( a_0, b_0 \right), b_1 \right) & = & r_B \left( \left. \tilde{a}_1, b _1 \right. \right) + \gamma \max _{b _2} Q_B ^{\ast} \left( \left( \tilde{a}_1, b_1 \right) , b _2 \right) , \label{Bellman1B}
\end{eqnarray}
\begin{equation}
\mathrm{where} \ \ \tilde{a} _1 = \operatorname*{argmax} _{a_1} Q_A ^{\ast} \left( \left( a_0, b_0 \right), a_1 \right), \ \tilde{b} _1 = \operatorname*{argmax} _{b_1} Q_B ^{\ast} \left( \left( a_0, b_0 \right), b_1 \right) . \label{greedy1}
\end{equation}
To derive Eq.~(\ref{greedy1}), we assume that agents employ the greedy method.

Assuming symmetry of the Q-function between the two agents,
\begin{equation}
Q_A ^{\ast} \left( \left( a_0, b_0 \right), a \right) = Q_B ^{\ast} \left( \left( b_0, a_0 \right), a \right) ,
\end{equation}
Eq.~(\ref{Bellman1B}) can be reduced to Eq.~(\ref{Bellman1A}). Therefore, we consider Eq.~(\ref{Bellman1A}) under the following conditions:
\begin{equation}
\tilde{b} _1 = \operatorname*{argmax} _{b_1} Q_A ^{\ast} \left( \left( b_0, a_0 \right), b_1 \right) . \label{opponent_strategy1}
\end{equation}
In the following discussions, we ignore the case in which the Q-functions for the two actions accidentally coincide: $Q_A ^{\ast} \left( \left( a, b \right), C \right) = Q_A ^{\ast} \left( \left( a, b \right), D \right)$, considering that such cases rarely occur.

\subsection{Limitation on the symmetric solution \label{limitation} }

In this section, we prove two lemmas that characterize the allowed pairwise ordering of the Q-function.

\begin{lem}
$Q_A ^{\ast} \left( \left( D, D \right), C \right) > Q_A ^{\ast} \left( \left( D, D \right), D \right)$ \label{lem1}
\end{lem}

\begin{pf}

If we assume $Q_A ^{\ast} \left( \left( D, D \right), C \right) < Q_A ^{\ast} \left( \left( D, D \right), D \right)$, then Bellman equation for $Q_A ^{\ast} \left( \left( D, D \right), a_1 \right)$ is expressed as:
\begin{equation}
Q_A ^{\ast} \left( \left( D, D \right), C \right) = S + \gamma \max _{a_2} Q_A ^{\ast} \left( \left( C, D \right), a_2 \right) , \label{QA_DD_C_w} 
\end{equation}
\begin{equation}
Q_A ^{\ast} \left( \left( D, D \right), D \right) = P + \gamma \max _{a_2} Q_A ^{\ast} \left( \left( D, D \right), a_2 \right) = P + \gamma Q_A ^{\ast} \left( \left( D, D \right), D \right) . \label{QA_DD_D_w}
\end{equation}
Solving Eq.~(\ref{QA_DD_D_w}), we obtain
\begin{equation}
Q_A ^{\ast} \left( \left( D, D \right), D \right) = \frac{P}{1 - \gamma} . \label{QA_DD_D_w2}
\end{equation}
Considering that the reward at each step has a minimum value, $P$, the following inequality is derived from the original definition of the Q-function:
\begin{eqnarray}
Q_A ^{\ast} \left( \left( a_0, b_0 \right), a_1 \right) & = & E \left[ \sum _{t=0} ^{\infty} \gamma ^t r_A \left( \left. a_{t+1}, b_{t+1} \right. \right) \middle| \left( a_0, b_0 \right), a_1 \right] \nonumber \\
& \geq & \sum _{t=0} ^{\infty} \gamma ^t P = \frac{P}{1 - \gamma }
\end{eqnarray}
Substituting this into Eqs.~(\ref{QA_DD_C_w}) and (\ref{QA_DD_D_w2}), we obtain the following inequality:
\begin{equation}
Q_A ^{\ast} \left( \left( D, D \right), C \right) - Q_A ^{\ast} \left( \left( D, D \right), D \right) \geq S + \frac{\gamma P}{1 - \gamma} - \frac{P}{1 - \gamma } = S - P > 0 .
\end{equation}
However, this contradicts the first assumption, $Q_A ^{\ast} \left( \left( D, D \right), C \right) < Q_A ^{\ast} \left( \left( D, D \right), D \right)$. Thus, $Q_A ^{\ast} \left( \left( D, D \right), C \right) > Q_A ^{\ast} \left( \left( D, D \right), D \right)$. \qed
\end{pf}

Using Lemma \ref{lem1}, the Bellman optimality equation for $Q_A ^{\ast} \left( \left( D, D \right), a \right)$ can be expressed as
\begin{equation}
Q_A ^{\ast} \left( \left( D, D \right), C \right) = R + \gamma \max _{a_2} Q_A ^{\ast} \left( \left( C, C \right), a_2 \right) , \label{QA_DD_C_t} 
\end{equation}
\begin{equation}
Q_A ^{\ast} \left( \left( D, D \right), D \right) = T + \gamma \max _{a_2} Q_A ^{\ast} \left( \left( D, C \right), a_2 \right) , \label{QA_DD_D_t}
\end{equation}

\begin{lem}
Only the following two order-relation patterns are permitted:
\begin{equation}
 Q_A ^{\ast} \left( \left( D, C \right), C \right) > Q_A ^{\ast} \left( \left( D, C \right), D \right) \ \mathrm{and} \ Q_A ^{\ast} \left( \left( C, D \right), C \right) > Q_A ^{\ast} \left( \left( C , D \right), D \right) \label{different_C}
\end{equation}
\begin{equation}
 Q_A ^{\ast} \left( \left( D, C \right), C \right) < Q_A ^{\ast} \left( \left( D, C \right), D \right) \ \mathrm{and} \ Q_A ^{\ast} \left( \left( C, D \right), C \right) < Q_A ^{\ast} \left( \left( C , D \right), D \right) \label{different_D}
\end{equation}
\end{lem}

\begin{pf}
Assuming $Q_A ^{\ast} \left( \left( D, C \right), C \right) < Q_A ^{\ast} \left( \left( D, C \right), D \right)$ and $Q_A ^{\ast} \left( \left( C, D \right), C \right) > Q_A ^{\ast} \left( \left( C , D \right), D \right)$, we obtain the following equations:
\begin{equation}
Q_A ^{\ast} \left( \left( D, C \right), C \right) = R + \gamma \max _{a_2} Q_A ^{\ast} \left( \left( C, C \right), a_2 \right) = Q_A ^{\ast} \left( \left( D, D \right), C \right) , \label{QA_DC_C_w} 
\end{equation}
\begin{equation}
Q_A ^{\ast} \left( \left( D, C \right), D \right) = T + \gamma \max _{a_2} Q_A ^{\ast} \left( \left( D, C \right), a_2 \right) = Q_A ^{\ast} \left( \left( D, D \right), D \right) , \label{QA_DC_D_w}
\end{equation}
Here, Eqs. (\ref{QA_DD_C_t}) and (\ref{QA_DD_D_t}) were used to derive these equations. Considering Eqs. (\ref{QA_DC_C_w}) and (\ref{QA_DC_D_w}), and Lemma \ref{lem1}, the following inequality holds:
\begin{equation}
Q_A ^{\ast} \left( \left( D, C \right), C \right) = Q_A ^{\ast} \left( \left( D, D \right), C \right) > Q_A ^{\ast} \left( \left( D, D \right), D \right) = Q_A ^{\ast} \left( \left( D, C \right), D \right) .
\end{equation}
However, this result contradicts the assumption $Q_A ^{\ast} \left( \left( D, C \right), C \right) < Q_A ^{\ast} \left( \left( D, C \right), D \right)$. Similarly, $Q_A ^{\ast} \left( \left( D, C \right), C \right) > Q_A ^{\ast} \left( \left( D, C \right), D \right)$ and $Q_A ^{\ast} \left( \left( C, D \right), C \right) < Q_A ^{\ast} \left( \left( C, D \right), D \right)$ produces the contradiction.

Therefore, Eqs.~(\ref{different_C}) and (\ref{different_D}) are the only permitted order-relation patterns. \qed
\end{pf}
 
\subsection{Explicit solutions \label{sym_solutions}}

Next, we derive explicit solutions for the four cases allowed by the lemmas in Sec.~\ref{limitation}. Each solution has the same form as in Ref.~\cite{Usui2021} because the structure of the Bellman equation does not depend directly on the relationship among $R, T, S$, and $P$, but on the assumed pairwise ordering of the Q-function.

In this discussion, we use the following expression:
\begin{equation}
\begin{aligned}
x_1 & \equiv R + \gamma \max _{a_2} Q_A ^{\ast} \left( \left( C, C \right), a_2 \right) , \\
x_2 & \equiv T + \gamma \max _{a_2} Q_A ^{\ast} \left( \left( D, C \right), a_2 \right) , \\
x_3 & \equiv S + \gamma \max _{a_2} Q_A ^{\ast} \left( \left( C, D \right), a_2 \right) , \\
x_4 & \equiv P + \gamma \max _{a_2} Q_A ^{\ast} \left( \left( D, D \right), a_2 \right) , \label{x_1234}
\end{aligned}
\end{equation}
for later convenience. This expression is primarily used in \ref{best_response}

\subsubsection{ The case $Q_A ^{\ast} \left( \left( C, C \right), C \right) > Q_A ^{\ast} \left( \left( C, C \right), D \right)$, $Q_A ^{\ast} \left( \left( C, D \right), C \right) > Q_A ^{\ast} \left( \left( C , D \right), D \right) $, $Q_A ^{\ast} \left( \left( D, C \right), C \right) > Q_A ^{\ast} \left( \left( D, C \right), D \right) $, and $Q_A ^{\ast} \left( \left( D, D \right), C \right) > Q_A ^{\ast} \left( \left( D, D \right), D \right) $ }

This corresponds to an All-C strategy.
The solution is given as
\begin{equation}
\begin{aligned}
x_1 &= Q_A ^{\ast} \left( \left( C, C \right), C \right) = Q_A ^{\ast} \left( \left( C, D \right), C \right) = Q_A ^{\ast} \left( \left( D, C \right), C \right) = Q_A ^{\ast} \left( \left( D, D \right), C \right) = \frac{1}{1 - \gamma} R   \\
x_2 &= Q_A ^{\ast} \left( \left( C, C \right), D \right) = Q_A ^{\ast} \left( \left( C, D \right), D \right) = Q_A ^{\ast} \left( \left( D, C \right), D \right) = Q_A ^{\ast} \left( \left( D, D \right), D \right) = T + \frac{\gamma}{1 - \gamma} R  \label{Q_AllC}
\end{aligned}
\end{equation}
This is always consistent with the condition of the harmony game.

\subsubsection{ The case $Q_A ^{\ast} \left( \left( C, C \right), C \right) > Q_A ^{\ast} \left( \left( C , C \right), D \right)$, $Q_A ^{\ast} \left( \left( C, D \right), C \right) < Q_A ^{\ast} \left( \left( C , D \right), D \right) $, $Q_A ^{\ast} \left( \left( D, C \right), C \right) < Q_A ^{\ast} \left( \left( D, C \right), D \right) $, and $Q_A ^{\ast} \left( \left( D, D \right), C \right) > Q_A ^{\ast} \left( \left( D, D \right), D \right) $ }

This corresponds to the WSLS strategy in the prisoner's dilemma \cite{Nowak1993}. In this study, we use the term ``WSLS'' for convenience, although the notions of ``win'' and ``lose'' in the harmony game differ from those in the prisoner's dilemma.
The solution is given as
\begin{equation}
\begin{aligned}
x_1 &= Q_A ^{\ast} \left( \left( C, C \right), C \right) = Q_A ^{\ast} \left( \left( D, D \right), C \right) =  \frac{1}{1 - \gamma} R  \\
x_2 &= Q_A ^{\ast} \left( \left( C, C \right), D \right) = Q_A ^{\ast} \left( \left( D, D \right), D \right) =  T + \gamma P + \frac{\gamma^2}{1 - \gamma} R ,  \\
x_3 &= Q_A ^{\ast} \left( \left( C, D \right), C \right) = Q_A ^{\ast} \left( \left( D, C \right), C \right) =  S + \gamma P + \frac{\gamma^2}{1 - \gamma} R ,  \\
x_4 &= Q_A ^{\ast} \left( \left( C, D \right), D \right) = Q_A ^{\ast} \left( \left( D, C \right), D \right) =  P + \frac{\gamma}{1 - \gamma} R , \label{Q_WSLS}
\end{aligned}
\end{equation}
This solution is consistent with the condition of the game when $\gamma > \frac{S-P}{R-P}$.

\subsubsection{ The case $Q_A ^{\ast} \left( \left( C, C \right), C \right) < Q_A ^{\ast} \left( \left( C , C \right), D \right)$, $Q_A ^{\ast} \left( \left( C, D \right), C \right) > Q_A ^{\ast} \left( \left( C , D \right), D \right) $, $Q_A ^{\ast} \left( \left( D, C \right), C \right) > Q_A ^{\ast} \left( \left( D, C \right), D \right) $, and $Q_A ^{\ast} \left( \left( D, D \right), C \right) > Q_A ^{\ast} \left( \left( D, D \right), D \right) $ }

In this case, the solution is given as
\begin{equation}
\begin{aligned}
x_1 &= Q_A ^{\ast} \left( \left( C, D \right), C \right) = Q_A ^{\ast} \left( \left( D, C \right), C \right) = Q_A ^{\ast} \left( \left( D, D \right), C \right) =  \frac{1}{1 - \gamma ^2} (R + \gamma P) ,  \\
x_2 &= Q_A ^{\ast} \left( \left( C, D \right), D \right) = Q_A ^{\ast} \left( \left( D, C \right), D \right) = Q_A ^{\ast} \left( \left( D, D \right), D \right) =  T + \frac{\gamma}{1 - \gamma ^2} (R + \gamma P) ,  \\
x_3 &= Q_A ^{\ast} \left( \left( C, C \right), C \right) =  S + \frac{\gamma}{1 - \gamma ^2} (R + \gamma P) ,  \\
x_4 &= Q_A ^{\ast} \left( \left( C, C \right), D \right) =  \frac{1}{1 - \gamma ^2} (P + \gamma R)  
\end{aligned}
\end{equation}
However, this solution contradicts the condition $S > P$.

\subsubsection{ The case $Q_A ^{\ast} \left( \left( C, C \right), C \right) < Q_A ^{\ast} \left( \left( C , C \right), D \right)$, $Q_A ^{\ast} \left( \left( C, D \right), C \right) < Q_A ^{\ast} \left( \left( C , D \right), D \right) $, $Q_A ^{\ast} \left( \left( D, C \right), C \right) < Q_A ^{\ast} \left( \left( D, C \right), D \right) $, and $Q_A ^{\ast} \left( \left( D, D \right), C \right) > Q_A ^{\ast} \left( \left( D, D \right), D \right) $ }

In this case, the solution is given as
\begin{equation}
\begin{aligned}
x_1 &= Q_A ^{\ast} \left( \left( D, D \right), C \right) =  \frac{1}{1 - \gamma^2} (R + \gamma P ) , \\
x_2 &= Q_A ^{\ast} \left( \left( D, D \right), D \right) =  T + \frac{\gamma}{1 - \gamma^2} (P + \gamma R ) , \\
x_3 &= Q_A ^{\ast} \left( \left( C, C \right), C \right) = Q_A ^{\ast} \left( \left( C, D \right), C \right) = Q_A ^{\ast} \left( \left( D, C \right), C \right) =  S + \frac{\gamma}{1 - \gamma^2} (P + \gamma R ) , \\
x_4 &= Q_A ^{\ast} \left( \left( C, C \right), D \right) = Q_A ^{\ast} \left( \left( C, D \right), D \right) = Q_A ^{\ast} \left( \left( D, C \right), D \right) =  \frac{1}{1 - \gamma^2} (P + \gamma R ) . \label{Q_FH}
\end{aligned}
\end{equation}
This solution is consistent with the condition of the game when $\gamma > \frac{S-P}{R-S}$.

In this strategy, an agent chooses $C$ if and only if both agents chose $D$ in the previous step. In the repeated prisoner's dilemma, the ``handshake mechanisms'', in which agents use particular initial action patterns as the signal for cooperation, have been proposed \cite{Li2011, Knight2018}. Compared to these mechanisms, this strategy seems to require frequent ``$D$'' handshakes from the opponent at an appropriate time as a condition of cooperation. Hence, in the following, we call this strategy the ``Frequent Handshake (FH) strategy''. Note that this strategy is termed the ``anti-Grim trigger'' in some studies \cite{Meylahn2025}, however, the definition of the term ``anti-Grim trigger'' is confusing because its definition varies across the literature.

The FH strategy may seem odd because it deliberately destroys cooperation even though defection yields a lower reward. Therefore, to illustrate the actual dynamics, we examine several cases in which agent A uses specific strategies while agent B's strategy is fixed as an FH strategy.
First, if agent A adopts the FH strategy, the sequence converges to the following limit cycle:
\begin{equation}
(C, C) \rightarrow (D, D) \rightarrow (C, C) \rightarrow ... . \label{CD_repeat}
\end{equation}
In this case, the long-term average of agent A's reward is $\frac{R + P}{2}$.
For agent A using the WSLS strategy, the sequence converges to the following limit cycle:
\begin{equation}
(C, C) \rightarrow (C, D) \rightarrow (D, D) \rightarrow (C, C) \rightarrow ... ,
\end{equation}
and the long-term average of agent A's reward is $\frac{R + S + P}{3}$.
If agent A emoplys the All-C strategy, the sequence converges to a fixed point,
\begin{equation}
(C, D) ,
\end{equation}
and the long-term average of agent A's reward is $S$. Among the three strategies, the FH strategy achieves the highest long-term average when 
\begin{equation}
\frac{R + P}{2} > S . \label{FH_condition_gamma1}
\end{equation}
The condition that the FH strategy appears, $\gamma > \frac{S-P}{R-S}$, converges to Eq.~(\ref{FH_condition_gamma1}) when $\gamma \rightarrow 1$. In addition, considering the following:
\begin{equation}
S = \min _b \max _a r_A \left( \left. a, b \right. \right) , 
\end{equation}
Eq.~(\ref{FH_condition_gamma1}) has the same form as the condition under which the folk theorem holds (for $\gamma$ sufficiently close to 1) \cite{Fudenberg1986}. %Indeed, the FH strategy penalizes deviations from the particular pattern specified in Eq.~(\ref{CD_repeat}), similar to the trigger strategies used in the proof of the folk theorem.

The best response of agent A under the fixed strategy of agent B is investigated in detail in \ref{best_response}.

\section{Numerical experiments \label{experiments}}

In this section, we numerically investigate the strategies adopted by actual reinforcement learning agents.
Here, Q-learning \cite{Watkins,Sutton} was employed as the learning algorithm, and the two agents independently updated their Q-functions $Q_A$ and $Q_B$. 
In other words, these functions were updated after the reward was calculated at each step using the following equation:
\begin{equation}
\begin{aligned}
Q_A \left( \left( a_{t-1}, b_{t-1} \right) , a_t \right) \leftarrow & Q_A \left( \left( a_{t-1}, b_{t-1} \right) , a_t \right) \\ & + \alpha \left[ r_A \left( a_{t}, b_{t} \right) + \gamma \max _{a'} Q_A \left( \left( a_{t}, b_{t} \right) , a' \right) - Q_A \left( \left( a_{t-1}, b_{t-1} \right) , a_t \right) \right] , \\
Q_B \left( \left( a_{t-1}, b_{t-1} \right) , a_t \right) \leftarrow & Q_B \left( \left( a_{t-1}, b_{t-1} \right) , a_t \right) \\ & + \alpha \left[ r_B \left( a_{t}, b_{t} \right) + \gamma \max _{b'} Q_B \left( \left( a_{t}, b_{t} \right) , b' \right) - Q_B \left( \left( a_{t-1}, b_{t-1} \right) , b_t \right) \right] , \label{Q_update}
\end{aligned}
\end{equation}
where $\alpha$ denotes the learning rate. The rewards for each joint action were set as $R=3$, $T=2$, $S=1$, and $P=0$.

One episode is truncated when the time step reaches $T_{\mathrm{ep} } = 100$, but the truncation is not treated as an episode termination when updating the Q-functions according to Eq.~(\ref{Q_update}). In the actual implementation, we realized this by recording the global time step $t$ and resetting the environment when $t$ was a multiple of $T_{\mathrm{ep} } = 100$; that is, when $t \mathbin{\%} T_{\mathrm{ep} } = t \mathbin{\%} 100 = 0$.
At the beginning of each episode, the observation of agents, $\left( a, b \right)$, is uniformly randomly chosen from $\left( C, C \right)$, $\left( C, D \right)$, $\left( D, C \right)$, and $\left( D, D \right)$.
In other words, the previous actions observed by the agents in the initial state are selected uniformly at random.
At each step, the agents choose their actions according to the $\epsilon$-greedy method.
For each agent, $i = A, B$, the initial value of $\epsilon _i$ is set to $\epsilon _{\mathrm{init} }$, and $\epsilon_i$ is subsequently decreased using the following equation:
\begin{equation}
\epsilon _i \leftarrow \max \left[ \epsilon _{\mathrm{min} } , (1 - \delta _{\epsilon}) \epsilon _i \right] ,
\end{equation}
We considered two values of $\epsilon _{\mathrm{init} }$: $\epsilon _{\mathrm{init} } = 1$ and $\epsilon _{\mathrm{init} } = 0.02$, whereas the other hyperparameters were fixed, as given in Table~\ref{hyperparameters} of \ref{simulation_details}. The detailed pseudocode is provided by Algorithm \ref{alg1} in this appendix.

Several initial Q-functions were also considered: $Q_{\mathrm{init} } = Q_{\mathrm{zero} }$, $Q_{\mathrm{WSLS} } $, and $Q_{\mathrm{FH} }$, where
\begin{equation}
Q_{\mathrm{zero} } \left( \left( a, b \right) , a' \right) = 0 \ \ \mathrm{for \ any \ } a, b, a' ,
\end{equation}
\begin{equation}
\begin{aligned}
Q_{\mathrm{WSLS} } \left( \left( C, C \right) , C \right) & = Q_{\mathrm{WSLS} } \left( \left( C, D \right) , D \right) = Q_{\mathrm{WSLS} } \left( \left( D, C \right) , D \right) = Q_{\mathrm{WSLS} } \left( \left( D, D \right) , C \right) = 1 , \\
Q_{\mathrm{WSLS} } \left( \left( C, C \right) , D \right) & = Q_{\mathrm{WSLS} } \left( \left( C, D \right) , C \right) = Q_{\mathrm{WSLS} } \left( \left( D, C \right) , C \right) = Q_{\mathrm{WSLS} } \left( \left( D, D \right) , D \right) = 0 , 
\end{aligned}
\end{equation}
and
\begin{equation}
\begin{aligned}
Q_{\mathrm{FH} } \left( \left( C, C \right) , D \right) & = Q_{\mathrm{FH} } \left( \left( C, D \right) , D \right) = Q_{\mathrm{FH} } \left( \left( D, C \right) , D \right) = Q_{\mathrm{FH} } \left( \left( D, D \right) , C \right) = 1 , \\
Q_{\mathrm{FH} } \left( \left( C, C \right) , C \right) & = Q_{\mathrm{FH} } \left( \left( C, D \right) , C \right) = Q_{\mathrm{FH} } \left( \left( D, C \right) , C \right) = Q_{\mathrm{FH} } \left( \left( D, D \right) , D \right) = 0 ,
\end{aligned}
\end{equation}
with 
\begin{equation}
\left. Q_A \left( \left( a, b \right) , a' \right) \right| _{t=0} = \left. Q_B \left( \left( b, a \right) , a' \right) \right| _{t=0} = Q_{\mathrm{init} } \left( \left( a, b \right) , a' \right) . \label{Q_initialization}
\end{equation}
As seen later, when $Q_{\mathrm{init} } = Q_{\mathrm{zero} }$, the agents adopt the All-C strategy with high probability. Therefore, we also investigated two other initial Q-functions, $Q_{\mathrm{WSLS} }$ and $Q_{\mathrm{FH} }$, under which agents tend to employ WSLS and FH strategies.

When tuning hyperparameters, we conducted 20 or 100 trials using a random seed different from that used when collecting the reported data. Fig.~\ref{Q-numerical} shows the time evolution of the Q-function in each case. In this figure, the solid curves and shaded areas show the median and interquartile intervals (both computed using the \texttt{numpy.quantile()}) over 1000 independent trials. The graphs of the cases with $Q_{\mathrm{init} } = Q_{\mathrm{WSLS} }$ or $Q_{\mathrm{FH} }$ and $\epsilon_{\mathrm{init} } = 1$ are omitted because their shapes are almost identical to the corresponding graph for $Q_{\mathrm{init} } = Q_{\mathrm{zero} }$ and $\epsilon_{\mathrm{init} } = 1$. Table~\ref{table_strategies} lists the number of trials (out of 1000) in which both agents selected each strategy. In this table, the learned strategy is identified by the ordering of the Q-functions at the end of the simulation, and the cases in which the two agents learned different strategies are excluded. According to the table, agents almost always learn the All-C strategy when $\epsilon_{\mathrm{init} } = 1$, whereas the learned strategy depends on $Q_{\mathrm{init} }$ when $\epsilon_{\mathrm{init} } = 0.02$. Specifically, agents learned the WSLS strategy when $Q_{\mathrm{init} } = Q_{\mathrm{WSLS} }$ and the FH strategy when $Q_{\mathrm{init} } = Q_{\mathrm{FH} }$ in most trials. Under $Q_{\mathrm{init} } = Q_{\mathrm{zero} }$ and $\epsilon_{\mathrm{init} } = 0.02$, the All-C strategy was learned most frequently. However, the WSLS and FH strategies were sometimes learned. Comparing the results for $Q_{\mathrm{init} } = Q_{\mathrm{zero} }$ with different values of $\epsilon_{\mathrm{init} }$, the frequency with which both agents learned the All-C strategy is higher when $\epsilon_{\mathrm{init} } = 1$ than when $\epsilon_{\mathrm{init} } = 0.02$. A possible explanation is that when the opponent behaves randomly, independent of the previous action, the appropriate strategy is always to choose $C$ as in the one-shot game.

To compare the simulation results with the theoretical values discussed in \ref{sym_solutions}, we plot the symmetric solutions of the Bellman optimality equation as dotted lines in Fig.~\ref{Q-numerical}. These lines correspond to the All-C strategy (Eq.~(\ref{Q_AllC})) in graphs (a) and (b), the WSLS strategy (Eq.~(\ref{Q_WSLS})) in (c), and the FH strategy (Eq.~(\ref{Q_FH})) in (d). Fig.~\ref{Q-numerical} indicates that the medians of the Q-functions converge to a symmetric solution corresponding to the strategies most frequently learned in the simulations. Note that in this simulation, strategies learned with low frequency behaved as outliers, which caused the average over the trials to deviate from the theoretical values. Therefore, we have plotted the median and interquartile intervals in Fig.~\ref{Q-numerical}.
\begin{figure}
\includegraphics[width = 12.0cm]{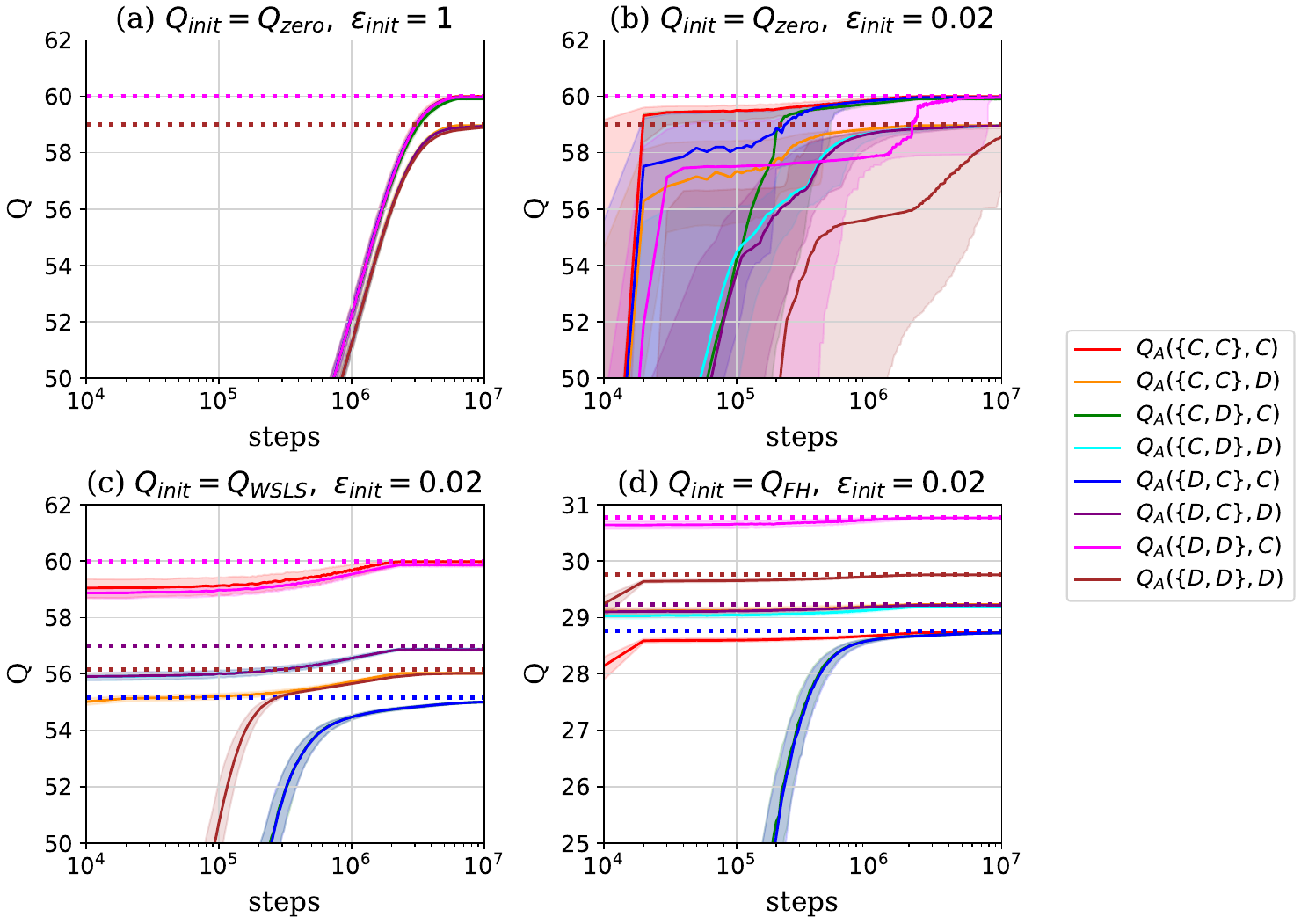}
\caption{Time evolution of the Q-function (of agent A)  with (a) $Q_{\mathrm{init} } = Q_{\mathrm{zero} }$ and $\epsilon_{\mathrm{init} } = 1$, (b) $Q_{\mathrm{init} } = Q_{\mathrm{zero} }$ and $\epsilon_{\mathrm{init} } = 0.02$, (c) $Q_{\mathrm{init} } = Q_{\mathrm{WSLS} }$ and $\epsilon_{\mathrm{init} } = 0.02$, and (d) $Q_{\mathrm{init} } = Q_{\mathrm{FH} }$ and $\epsilon_{\mathrm{init} } = 0.02$. The solid curves and shaded areas of each color represent the median and interquartile intervals of each component $Q_A \left( \left( a, b \right) , a' \right) $. The dotted lines are the corresponding components of the symmetric solution of the Bellman optimality equation, yielding (a), (b) the All-C strategy (Eq.~(\ref{Q_AllC})); (c) the WSLS strategy (Eq.~(\ref{Q_WSLS})); and (d) the FH strategy (Eq.~(\ref{Q_FH})). }
\label{Q-numerical}
\end{figure}

\begin{table*}
\centering
\hspace{-4.0cm}
\begin{tabular}{cccc}
initial condition & All-C & WSLS & FH \\ \hline\hline
$Q_{\mathrm{init} } = Q_{\mathrm{zero} }$, \ $\epsilon_{\mathrm{init} } = 1$ & 1000 & 0 & 0 \\ \hline
$Q_{\mathrm{init} } = Q_{\mathrm{WSLS} }$, \ $\epsilon_{\mathrm{init} } = 1$ & 1000 & 0 & 0 \\ \hline
$Q_{\mathrm{init} } = Q_{\mathrm{FH} }$, \ $\epsilon_{\mathrm{init} } = 1$ & 1000 & 0 & 0 \\ \hline
$Q_{\mathrm{init} } = Q_{\mathrm{zero} }$, \ $\epsilon_{\mathrm{init} } = 0.02$ & 823 & 64 & 13 \\ \hline
$Q_{\mathrm{init} } = Q_{\mathrm{WSLS} }$, \ $\epsilon_{\mathrm{init} } = 0.02$ & 1 & 999 & 0 \\ \hline
$Q_{\mathrm{init} } = Q_{\mathrm{FH} }$, \ $\epsilon_{\mathrm{init} } = 0.02$ & 0 & 7 &993 \\ \hline
\end{tabular}
\vspace{1.0mm}
\caption{Number of trials (out of 1000) in which both agents learned each strategy. }
\label{table_strategies}
\end{table*}

%\subsection{fixed strategy}
%We first simulate agent A's learning while fixing agent B's strategy, to certificate the action which strategy yields the highest score against a fixed opponent.

\section{Summary \label{Summary} }

In this study, we analyzed the Bellman optimality equation for a repeated harmony game and identified three symmetric solutions: All-C, WSLS, and a third strategy, which we referred to as the Frequent Handshake (FH). In the FH strategy, the agent chooses $C$ (cooperation) only when both it and the opponent have chosen $D$ (defection) in the previous step. If both agents employ this strategy, they repeat $\left( C, C \right)$ and $\left( D, D \right)$ alternately and punish deviations from this pattern by choosing $D$.

Numerical experiments using Q-learning showed that agents frequently learn WSLS or FH strategies if the exploration noise is sufficiently weak and the initial Q-function is biased toward these strategies. Our original motivation was to promote cooperation in the social dilemma game by introducing additional rewards, and the harmony game was considered as an ideal case in which these additional rewards completely resolve the dilemma at least in the one-shot game. The fact that such nontrivial strategies exist even in the harmony game makes the design of the additional reward complex. In particular, the FH strategy frequently chooses $D$ and consequently receives less reward than the All-C and WSLS strategies. 
If frequent defections were observed in more complex environments or under more sophisticated reinforcement learning algorithms, it would be difficult to determine whether they reflected a non-trivial strategy learned by the agents or were merely due to insufficient relaxation or algorithmic issues. 

The source code for this work is uploaded to \url{https://github.com/Hisato-Komatsu/Repeated_Harmony}.

\section*{Acknowledgements}
We would like to thank Editage (www.editage.jp) for English language editing.

\section*{Declaration of generative AI and AI-assisted technologies in the manuscript preparation process}

During the preparation of this work, the author used DeepL, Google Translate, Microsoft Copilot, and ChatGPT for English language editing, translation, research-related queries, and source code review and improvement suggestions.
After using these tools, the author reviewed and edited the content as needed and takes full responsibility for the content of the publication.

\appendix
%\section{appendix}

\section{Details of the numerical experiment \label{simulation_details} }

This section describes the pseudocode (Algorithm \ref{alg1}), and lists the hyperparameters (Table \ref{hyperparameters}) for reproducibility.
%\begin{figure}
\begin{algorithm}[H]
    \caption{Algorithm used for the calculation }
    \label{alg1}
\begin{algorithmic}[1]
    \State Initialize $Q_A$ and $Q_B$ using Eq.~(\ref{Q_initialization})
    \State $\epsilon _A, \epsilon _B \leftarrow \epsilon _{\mathrm{init} }$
    \State Set the initial observation $\left( a_0 , b_0 \right)$ randomly
    \For {$t \gets 1, ... , T_{\mathrm{max} } $}
        %\Repeat
        \State Choose agents' action $a_t$ and $b_t$ using $\epsilon$-greedy method with $\epsilon = \epsilon _A$ and $\epsilon = \epsilon _B$, respectively
        \State Calculate the rewards $r_A$ and $r_B$
        \State Update $Q_A$ and $Q_B$ using Eq.~(\ref{Q_update})
        \State $\epsilon _A \gets \max \left[ \epsilon_{\mathrm{min} }, (1 - \delta_{\epsilon}) \epsilon _A \right]$
        \State $\epsilon _B \gets \max \left[ \epsilon_{\mathrm{min} }, (1 - \delta_{\epsilon}) \epsilon _B \right]$
        \If{$t \mathbin{\%} T_{\mathrm{ep} } = 0$ } 
            \State truncate the episode and reset the environment by randomly reselecting $\left( a_t , b_t \right)$
        \EndIf
   \EndFor
\end{algorithmic}
\end{algorithm}
%\end{figure}

\begin{table*}
\centering
\hspace{-4.0cm}
\begin{tabular}{ccc}
 character & meaning & value \\ \hline\hline
$\gamma$ & discount factor & 0.95 \\ \hline
$\delta _{\epsilon}$ & decaying rate of $\epsilon$ per one step & $1 \times10^{-6}$  \\ \hline
$\epsilon _{\mathrm{min} }$ & minimum value of $\epsilon$ & 0.002  \\ \hline
$\alpha$ & learning rate & 0.1 \\ \hline
$T_{\mathrm{ep} }$ & episode length & 100 \\ \hline
$T_{\mathrm{max} }$ & whole time steps & $10^7$ \\ \hline
$R$ & -- & 3 \\ \hline
$T$ & -- & 2 \\ \hline
$S$ & -- & 1 \\ \hline
$P$ & -- & 0 \\ \hline

\end{tabular}
\vspace{1.0mm}
\caption{Hyperparameters of this study. }
\label{hyperparameters}
\end{table*}

\section{Optimal strategy against fixed strategy \label{best_response} }

To gain insights into the stability of a strategy, it is common practice to consider an agent's learning dynamics while the opponent's strategy is fixed \cite{Meylahn2022, Meylahn2025}. Therefore, we briefly discuss agent A's learning when agent B's strategy is fixed. Assuming that agent B follows a deterministic policy $\pi _B$, the Bellman optimality equation for agent A is given by Eq.~(\ref{Bellman1A}) with
\begin{equation}
\tilde{b} _1 = \pi _B \left( \left. a_0, b_0 \right. \right) ,
\end{equation}
where $\pi _B \left( \left. a,b \right. \right)$ denotes agent B's action when the previous action pair $\left( a,b \right)$ is given.

To represent agent B's strategy, we use the following vector:
\begin{equation}
\left( \pi _B \left( \left. C,C \right. \right) , \pi _B \left( \left. C,D \right. \right) , \pi _B \left( \left. D,C \right. \right) , \pi _B \left( \left. D,D \right. \right) \right) .
\end{equation}
For example, the WSLS strategy is expressed as $(C,D,D,C)$. First, we discuss the case in which agent B takes the strategies that genetate the symmetric solutions discussed in Sec.~\ref{sym_solutions}, WSLS, FH, and All-C.

\subsection{Optimal strategy against the WSLS strategy (The case $(C,D,D,C)$)}

When agent B's strategy is fixed as the WSLS strategy, the Bellman optimality equation is expressed as follows: 
\begin{equation}
\begin{aligned}
Q_A ^{\ast} \left( \left( C, C \right), C \right) & = Q_A ^{\ast} \left( \left( D, D \right), C \right) = R + \gamma \max _{a_2} Q_A ^{\ast} \left( \left( C, C \right), a_2 \right) , \\
Q_A ^{\ast} \left( \left( C, C \right), D \right) & = Q_A ^{\ast} \left( \left( D, D \right), D \right) = T + \gamma \max _{a_2} Q_A ^{\ast} \left( \left( D, C \right), a_2 \right) , \\
Q_A ^{\ast} \left( \left( C, D \right), C \right) & = Q_A ^{\ast} \left( \left( D, C \right), C \right) = S + \gamma \max _{a_2} Q_A ^{\ast} \left( \left( C, D \right), a_2 \right) , \\
Q_A ^{\ast} \left( \left( C, D \right), D \right) & = Q_A ^{\ast} \left( \left( D, C \right), D \right) = P + \gamma \max _{a_2} Q_A ^{\ast} \left( \left( D, D \right), a_2 \right) . \label{Bellman_WSLS}
\end{aligned}
\end{equation}
Using Eq.~(\ref{x_1234}), Eq.~(\ref{Bellman_WSLS}) is transformed into
\begin{equation}
\begin{aligned}
x_1 & = R + \gamma \max \left( x_1 , x_2 \right) , \\
x_2 & = T + \gamma \max \left( x_3 , x_4 \right) , \\
x_3 & = S + \gamma \max \left( x_3 , x_4 \right) , \\
x_4 & = P + \gamma \max \left( x_1 , x_2 \right) . 
\end{aligned}
\end{equation}
By solving these equations, we obtain the same form as Eq. (\ref{Q_WSLS}) with $x_1 > x_2$ and $x_3 < x_4$ when $\gamma > \frac{S- P}{R-P}$, and 
\begin{equation}
\begin{aligned}
x_1 & = Q_A ^{\ast} \left( \left( C, C \right), C \right) = Q_A ^{\ast} \left( \left( D, D \right), C \right) = \frac{1}{1 - \gamma} R ,  \\
x_2 & = Q_A ^{\ast} \left( \left( C, C \right), D \right) = Q_A ^{\ast} \left( \left( D, D \right), D \right) = T + \frac{\gamma}{1 - \gamma} S ,  \\
x_3 & = Q_A ^{\ast} \left( \left( C, D \right), C \right) = Q_A ^{\ast} \left( \left( D, C \right), C \right) = \frac{1}{1 - \gamma} S ,  \\
x_4 & = Q_A ^{\ast} \left( \left( C, D \right), D \right) = Q_A ^{\ast} \left( \left( D, C \right), D \right) = P + \frac{\gamma}{1 - \gamma} R . 
\end{aligned}
\end{equation}
with $x_1 > x_2$ and $x_3 > x_4$ when $\gamma < \frac{S- P}{R-P}$. Seeing the ordering of the Q-function values, the former and latter solutions correspond to the WSLS and All-C strategies, respectively.

\subsection{Optimal strategy against the FH strategy (The case $(D,D,D,C)$)}

In this case, the Bellman optimality equation is expressed as
\begin{equation}
\begin{aligned}
Q_A ^{\ast} \left( \left( D, D \right), C \right) & = R + \gamma \max _{a_2} Q_A ^{\ast} \left( \left( C, C \right), a_2 \right) , \\
Q_A ^{\ast} \left( \left( D, D \right), D \right) & = T + \gamma \max _{a_2} Q_A ^{\ast} \left( \left( D, C \right), a_2 \right) , \\
Q_A ^{\ast} \left( \left( C, C \right), C \right) & = Q_A ^{\ast} \left( \left( C, D \right), C \right) = Q_A ^{\ast} \left( \left( D, C \right), C \right) = S + \gamma \max _{a_2} Q_A ^{\ast} \left( \left( C, D \right), a_2 \right) , \\
Q_A ^{\ast} \left( \left( C, C \right), D \right) & = Q_A ^{\ast} \left( \left( C, D \right), D \right) = Q_A ^{\ast} \left( \left( D, C \right), D \right) = P + \gamma \max _{a_2} Q_A ^{\ast} \left( \left( D, D \right), a_2 \right) .
\end{aligned}
\end{equation}
By solving these equations, we obtain the same form as Eq.~(\ref{Q_FH}) with $x_1 > x_2$ and $x_3 < x_4$ when $\gamma > \frac{S- P}{R-S}$ and 
\begin{equation}
\begin{aligned}
x_1 & = Q_A ^{\ast} \left( \left( D, D \right), C \right) = R + \frac{\gamma}{1 - \gamma} S , \\
x_2 & = Q_A ^{\ast} \left( \left( D, D \right), D \right) = T + \frac{\gamma}{1 - \gamma} S , \\
x_3 & = Q_A ^{\ast} \left( \left( C, C \right), C \right) = Q_A ^{\ast} \left( \left( C, D \right), C \right) = Q_A ^{\ast} \left( \left( D, C \right), C \right) = \frac{1}{1 - \gamma} S , \\
x_4 & = Q_A ^{\ast} \left( \left( C, C \right), D \right) = Q_A ^{\ast} \left( \left( C, D \right), D \right) = Q_A ^{\ast} \left( \left( D, C \right), D \right) = P + \gamma R + \frac{\gamma ^2}{1 - \gamma} S , 
\end{aligned}
\end{equation}
with $x_1 > x_2$ and $x_3 > x_4$ when $\gamma < \frac{S- P}{R-S}$. The former and latter solutions correspond to the FH and All-C strategies, respectively.

\subsection{Optimal strategy against the All-C strategy (The case $(C, C, C, C)$)}

In this case, by solving the Bellman optimality equation, as in the two cases above, we obtain the same form as Eq.~(\ref{Q_AllC}), with $x_1 > x_2$.
This solution corresponds to the All-C strategy. In other words, when agent B adopts the All-C strategy, the optimal strategy for agent A is also the All-C one.

Note that, for All-C, WSLS, and FH, the condition under which the optimal strategy against each strategy coincides with the respective strategy is the same as that under which the symmetric solution discussed in Sec.~\ref{sym_solutions} emerges. This implies that the strategy providing the symmetric solution represents a Nash equilibrium, in which neither of the two agents can unilaterally change their strategy to increase their cumulative reward, when the strategy set is restricted to memory-1 strategies.

\subsection{Optimal strategies against other fixed strategies }

We also consider the case in which the strategy of agent B is fixed as one of 13 (deterministic) strategies that do not correspond to symmetric solutions. 

\subsubsection{The case $(C, C, C, D)$ }

In this case, we obtain
\begin{equation}
\begin{aligned}
x_1 & = Q_A ^{\ast} \left( \left( C, C \right), C \right) = Q_A ^{\ast} \left( \left( C, D \right), C \right) = Q_A ^{\ast} \left( \left( D, C \right), C \right) = \frac{1}{1 - \gamma} R , \\
x_2 & = Q_A ^{\ast} \left( \left( C, C \right), D \right) = Q_A ^{\ast} \left( \left( C, D \right), D \right) = Q_A ^{\ast} \left( \left( D, C \right), D \right) = T + \frac{\gamma}{1 - \gamma} R , \\
x_3 & = Q_A ^{\ast} \left( \left( D, D \right), C \right) = S + \frac{\gamma}{1 - \gamma} R , \\
x_4 & = Q_A ^{\ast} \left( \left( D, D \right), D \right) = P + \gamma S + \frac{\gamma^2}{1 - \gamma} R ,
\end{aligned}
\end{equation}
with $x_1 > x_2$ and $x_3 > x_4$. This solution corresponds to the All-C strategy.

\subsubsection{The case $(C, C, D, C)$ }

In this case, we obtain
\begin{equation}
\begin{aligned}
x_1 & = Q_A ^{\ast} \left( \left( C, C \right), C \right) = Q_A ^{\ast} \left( \left( C, D \right), C \right) = Q_A ^{\ast} \left( \left( D, D \right), C \right) = \frac{1}{1 - \gamma} R , \\
x_2 & = Q_A ^{\ast} \left( \left( C, C \right), D \right) = Q_A ^{\ast} \left( \left( C, D \right), D \right) = Q_A ^{\ast} \left( \left( D, D \right), D \right) = T + \gamma S + \frac{\gamma ^2}{1 - \gamma} R , \\
x_3 & = Q_A ^{\ast} \left( \left( D, C \right), C \right) = S + \frac{\gamma}{1 - \gamma} R , \\
x_4 & = Q_A ^{\ast} \left( \left( D, C \right), D \right) = P + \frac{\gamma}{1 - \gamma} R ,
\end{aligned}
\end{equation}
with $x_1 > x_2$ and $x_3 > x_4$. This solution corresponds to the All-C strategy.

\subsubsection{The case $(C, D, C, C)$ }

In this case, we obtain
\begin{equation}
\begin{aligned}
x_1 & = Q_A ^{\ast} \left( \left( C, C \right), C \right) = Q_A ^{\ast} \left( \left( D, C \right), C \right) = Q_A ^{\ast} \left( \left( D, D \right), C \right) = \frac{1}{1 - \gamma} R , \\
x_2 & = Q_A ^{\ast} \left( \left( C, C \right), D \right) = Q_A ^{\ast} \left( \left( D, C \right), D \right) = Q_A ^{\ast} \left( \left( D, D \right), D \right) = T + \frac{\gamma}{1 - \gamma} R , \\
x_3 & = Q_A ^{\ast} \left( \left( C, D \right), C \right) = \frac{1}{1 - \gamma} S , \\
x_4 & = Q_A ^{\ast} \left( \left( C, D \right), D \right) = P + \frac{\gamma}{1 - \gamma} R ,
\end{aligned}
\end{equation}
with $x_1 > x_2$ and $x_3 > x_4$, when $\gamma < \frac{S-P}{R-P}$, and
\begin{equation}
\begin{aligned}
x_1 & = Q_A ^{\ast} \left( \left( C, C \right), C \right) = Q_A ^{\ast} \left( \left( D, C \right), C \right) = Q_A ^{\ast} \left( \left( D, D \right), C \right) = \frac{1}{1 - \gamma} R , \\
x_2 & = Q_A ^{\ast} \left( \left( C, C \right), D \right) = Q_A ^{\ast} \left( \left( D, C \right), D \right) = Q_A ^{\ast} \left( \left( D, D \right), D \right) = T + \frac{\gamma}{1 - \gamma} R , \\
x_3 & = Q_A ^{\ast} \left( \left( C, D \right), C \right) = S + \gamma P + \frac{\gamma^2}{1 - \gamma} R , \\
x_4 & = Q_A ^{\ast} \left( \left( C, D \right), D \right) = P + \frac{\gamma}{1 - \gamma} R ,
\end{aligned}
\end{equation}
with $x_1 > x_2$ and $x_3 < x_4$, when $\gamma > \frac{S-P}{R-P}$. The former and latter solutions correspond to the All-C and $(C, C, D, C)$ strategies.

\subsubsection{The case $(D, C, C, C)$ }

In this case, we obtain
\begin{equation}
\begin{aligned}
x_1 & = Q_A ^{\ast} \left( \left( C, D \right), C \right) = Q_A ^{\ast} \left( \left( D, C \right), C \right) = Q_A ^{\ast} \left( \left( D, D \right), C \right) = \frac{1}{1 - \gamma ^2} (R + \gamma S) , \\
x_2 & = Q_A ^{\ast} \left( \left( C, D \right), D \right) = Q_A ^{\ast} \left( \left( D, C \right), D \right) = Q_A ^{\ast} \left( \left( D, D \right), D \right) = T + \frac{\gamma}{1 - \gamma ^2} (R + \gamma S) , \\
x_3 & = Q_A ^{\ast} \left( \left( C, C \right), C \right) = \frac{1}{1 - \gamma ^2} (S + \gamma R) , \\
x_4 & = Q_A ^{\ast} \left( \left( C, C \right), D \right) = P + \frac{\gamma}{1 - \gamma ^2} (R + \gamma S) ,
\end{aligned}
\end{equation}
with $x_1 > x_2$ and $x_3 > x_4$, when $\gamma (T - S) < R - T$, and
\begin{equation}
\begin{aligned}
x_1 & = Q_A ^{\ast} \left( \left( C, D \right), C \right) = Q_A ^{\ast} \left( \left( D, C \right), C \right) = Q_A ^{\ast} \left( \left( D, D \right), C \right) = R + \gamma S + \frac{\gamma ^2}{1 - \gamma} T , \\
x_2 & = Q_A ^{\ast} \left( \left( C, D \right), D \right) = Q_A ^{\ast} \left( \left( D, C \right), D \right) = Q_A ^{\ast} \left( \left( D, D \right), D \right) = \frac{1}{1 - \gamma} T , \\
x_3 & = Q_A ^{\ast} \left( \left( C, C \right), C \right) = S + \frac{\gamma}{1 - \gamma} T , \\
x_4 & = Q_A ^{\ast} \left( \left( C, C \right), D \right) = P + \frac{\gamma}{1 - \gamma} T ,
\end{aligned}
\end{equation}
with $x_1 < x_2$ and $x_3 > x_4$, when $\gamma (T - S) > R - T$. The former and latter solutions correspond to the All-C and Grim strategies, respectively. Here, following Refs.~\cite{Usui2021,Meylahn2022,Meylahn2025}, we use the term ``Grim strategy'' to refer to the 1-memory strategy that chooses $C$ if and only if the previous joint action is $\left( C, C \right)$, i.e., the $(C,D,D,D)$ strategy. This terminology differs from the conventional Grim Trigger strategy in which a deviation triggers permanent punishment \cite{Harrington1991,McGillivray2000}.

\subsubsection{The case $(C, C, D, D)$ }

In this case, agent B responds with the Tit-for-Tat (TFT) strategy \cite{Axelrod1981}. By solving the Bellman optimality equation, we obtain
\begin{equation}
\begin{aligned}
x_1 & = Q_A ^{\ast} \left( \left( C, C \right), C \right) = Q_A ^{\ast} \left( \left( C, D \right), C \right) = \frac{1}{1 - \gamma} R , \\
x_2 & = Q_A ^{\ast} \left( \left( C, C \right), D \right) = Q_A ^{\ast} \left( \left( C, D \right), D \right) = T + \gamma S + \frac{\gamma ^2}{1 - \gamma} R , \\
x_3 & = Q_A ^{\ast} \left( \left( D, C \right), C \right) = Q_A ^{\ast} \left( \left( D, D \right), C \right) = S + \frac{\gamma}{1 - \gamma} R , \\
x_4 & = Q_A ^{\ast} \left( \left( D, C \right), D \right) = Q_A ^{\ast} \left( \left( D, D \right), D \right) = P + \gamma S + \frac{\gamma ^2}{1 - \gamma} R ,
\end{aligned}
\end{equation}
with $x_1 > x_2$ and $x_3 > x_4$. This solution corresponds to the All-C strategy.

\subsubsection{The case $(C, D, C, D)$}

In this case, we obtain
\begin{equation}
\begin{aligned}
x_1 & = Q_A ^{\ast} \left( \left( C, C \right), C \right) = Q_A ^{\ast} \left( \left( D, C \right), C \right) = \frac{1}{1 - \gamma} R , \\
x_2 & = Q_A ^{\ast} \left( \left( C, C \right), D \right) = Q_A ^{\ast} \left( \left( D, C \right), D \right) = T + \frac{\gamma}{1 - \gamma} R , \\
x_3 & = Q_A ^{\ast} \left( \left( C, D \right), C \right) = Q_A ^{\ast} \left( \left( D, D \right), C \right) = \frac{1}{1 - \gamma} S , \\
x_4 & = Q_A ^{\ast} \left( \left( C, D \right), D \right) = Q_A ^{\ast} \left( \left( D, D \right), D \right) = P + \frac{\gamma}{1 - \gamma} S , 
\end{aligned}
\end{equation}
with $x_1 > x_2$ and $x_3 > x_4$. This solution corresponds to the All-C strategy.

\subsubsection{The case $(D, C, C, D)$ }

In this case, we obtain
\begin{equation}
\begin{aligned}
x_1 & = Q_A ^{\ast} \left( \left( C, D \right), C \right) = Q_A ^{\ast} \left( \left( D, C \right), C \right) = \frac{1}{1 - \gamma^2} (R + \gamma S) , \\
x_2 & = Q_A ^{\ast} \left( \left( C, D \right), D \right) = Q_A ^{\ast} \left( \left( D, C \right), D \right) = T + \frac{\gamma}{1 - \gamma^2} (R + \gamma S) , \\
x_3 & = Q_A ^{\ast} \left( \left( C, C \right), C \right) =  Q_A ^{\ast} \left( \left( D, D \right), C \right) = \frac{1}{1 - \gamma^2} (S + \gamma R) , \\
x_4 & = Q_A ^{\ast} \left( \left( C, C \right), D \right) = Q_A ^{\ast} \left( \left( D, D \right), D \right) = P + \frac{\gamma}{1 - \gamma^2} (S + \gamma R) , 
\end{aligned}
\end{equation}
with $x_1 > x_2$ and $x_3 > x_4$, when $R - T > \gamma (T - S)$, and 
\begin{equation}
\begin{aligned}
x_1 & = Q_A ^{\ast} \left( \left( C, D \right), C \right) = Q_A ^{\ast} \left( \left( D, C \right), C \right) = R + \gamma S + \frac{\gamma ^2}{1 - \gamma} T , \\
x_2 & = Q_A ^{\ast} \left( \left( C, D \right), D \right) = Q_A ^{\ast} \left( \left( D, C \right), D \right) = \frac{1}{1 - \gamma} T , \\
x_3 & = Q_A ^{\ast} \left( \left( C, C \right), C \right) =  Q_A ^{\ast} \left( \left( D, D \right), C \right) = S + \frac{\gamma}{1 - \gamma} T , \\
x_4 & = Q_A ^{\ast} \left( \left( C, C \right), D \right) = Q_A ^{\ast} \left( \left( D, D \right), D \right) = P + \gamma S + \frac{\gamma ^2}{1 - \gamma} T , 
\end{aligned}
\end{equation}
with $x_1 < x_2$ and $x_3 > x_4$, when $R - T < \gamma (T - S)$. The former and latter solutions correspond to the All-C and WSLS strategies.

\subsubsection{The case $(D, C, D, C)$}

In this case, we obtain
\begin{equation}
\begin{aligned}
x_1 & = Q_A ^{\ast} \left( \left( C, D \right), C \right) = Q_A ^{\ast} \left( \left( D, D \right), C \right) = \frac{1}{1 - \gamma^2} (R + \gamma S) , \\
x_2 & = Q_A ^{\ast} \left( \left( C, D \right), D \right) = Q_A ^{\ast} \left( \left( D, D \right), D \right) = T + \frac{\gamma}{1 - \gamma^2} (S + \gamma R) , \\
x_3 & = Q_A ^{\ast} \left( \left( C, C \right), C \right) = Q_A ^{\ast} \left( \left( D, C \right), C \right) = \frac{1}{1 - \gamma^2} (S + \gamma R) , \\
x_4 & = Q_A ^{\ast} \left( \left( C, C \right), D \right) = Q_A ^{\ast} \left( \left( D, C \right), D \right) = P + \frac{\gamma}{1 - \gamma^2} (R + \gamma S) , 
\end{aligned}
\end{equation}
with $x_1 > x_2$ and $x_3 > x_4$. This solution corresponds to the All-C strategy.

\subsubsection{The case $(D, D, C, C)$ }

In this case, we obtain
\begin{equation}
\begin{aligned}
x_1 & = Q_A ^{\ast} \left( \left( D, C \right), C \right) = Q_A ^{\ast} \left( \left( D, D \right), C \right) = R + \frac{\gamma}{1 - \gamma} S , \\
x_2 & = Q_A ^{\ast} \left( \left( D, C \right), D \right) = Q_A ^{\ast} \left( \left( D, D \right), D \right) = T + \gamma R + \frac{\gamma ^2}{1 - \gamma} S , \\
x_3 & = Q_A ^{\ast} \left( \left( C, C \right), C \right) = Q_A ^{\ast} \left( \left( C, D \right), C \right) = \frac{1}{1 - \gamma} S , \\
x_4 & = Q_A ^{\ast} \left( \left( C, C \right), D \right) = Q_A ^{\ast} \left( \left( C, D \right), D \right) = P + \gamma R + \frac{\gamma ^2}{1 - \gamma} S ,
\end{aligned}
\end{equation}
with $x_1 > x_2$ and $x_3 > x_4$, when $T - \gamma S < ( 1 - \gamma ) R$ and $S > \frac{P + \gamma R}{1 + \gamma}$, 
\begin{equation}
\begin{aligned}
x_1 & = Q_A ^{\ast} \left( \left( D, C \right), C \right) = Q_A ^{\ast} \left( \left( D, D \right), C \right) = \frac{1}{1 - \gamma^2} (R + \gamma P) , \\
x_2 & = Q_A ^{\ast} \left( \left( D, C \right), D \right) = Q_A ^{\ast} \left( \left( D, D \right), D \right) = T + \frac{\gamma}{1 - \gamma^2} (R + \gamma P) , \\
x_3 & = Q_A ^{\ast} \left( \left( C, C \right), C \right) = Q_A ^{\ast} \left( \left( C, D \right), C \right) = S + \frac{\gamma}{1 - \gamma^2} (P + \gamma R) , \\
x_4 & = Q_A ^{\ast} \left( \left( C, C \right), D \right) = Q_A ^{\ast} \left( \left( C, D \right), D \right) = \frac{1}{1 - \gamma^2} (P + \gamma R) ,
\end{aligned}
\end{equation}
with $x_1 > x_2$ and $x_3 < x_4$, when $T < \frac{R + \gamma P}{1 + \gamma}$ and $S < \frac{P + \gamma R}{1 + \gamma}$, 
\begin{equation}
\begin{aligned}
x_1 & = Q_A ^{\ast} \left( \left( D, C \right), C \right) = Q_A ^{\ast} \left( \left( D, D \right), C \right) = R + \frac{\gamma}{1 - \gamma} S , \\
x_2 & = Q_A ^{\ast} \left( \left( D, C \right), D \right) = Q_A ^{\ast} \left( \left( D, D \right), D \right) = \frac{1}{1 - \gamma} T , \\
x_3 & = Q_A ^{\ast} \left( \left( C, C \right), C \right) = Q_A ^{\ast} \left( \left( C, D \right), C \right) = \frac{1}{1 - \gamma} S , \\
x_4 & = Q_A ^{\ast} \left( \left( C, C \right), D \right) = Q_A ^{\ast} \left( \left( C, D \right), D \right) = P + \frac{\gamma}{1 - \gamma} T ,
\end{aligned}
\end{equation}
with $x_1 < x_2$ and $x_3 > x_4$, when $T - \gamma S > ( 1 - \gamma ) R$ and $S - \gamma T > ( 1 - \gamma ) P$, and
\begin{equation}
\begin{aligned}
x_1 & = Q_A ^{\ast} \left( \left( D, C \right), C \right) = Q_A ^{\ast} \left( \left( D, D \right), C \right) = R + \gamma P + \frac{\gamma ^2}{1 - \gamma} T , \\
x_2 & = Q_A ^{\ast} \left( \left( D, C \right), D \right) = Q_A ^{\ast} \left( \left( D, D \right), D \right) = \frac{1}{1 - \gamma} T , \\
x_3 & = Q_A ^{\ast} \left( \left( C, C \right), C \right) = Q_A ^{\ast} \left( \left( C, D \right), C \right) = S + \gamma P + \frac{\gamma ^2}{1 - \gamma} T , \\
x_4 & = Q_A ^{\ast} \left( \left( C, C \right), D \right) = Q_A ^{\ast} \left( \left( C, D \right), D \right) = P + \frac{\gamma}{1 - \gamma} T ,
\end{aligned}
\end{equation}
with $x_1 < x_2$ and $x_3 < x_4$, when $T > \frac{R + \gamma P}{1 + \gamma}$ and $S - \gamma T < ( 1 - \gamma ) P$.
These solutions correspond to the All-C, $(D, C, D, C)$, $(C, D, C, D)$, and All-D strategies.

\subsubsection{The case $(C, D, D, D)$ }

In this case, agent B's stragtegy corresponds to the Grim one. Solving the Bellman optimality equation, we obtain
\begin{equation}
\begin{aligned}
x_1 & = Q_A ^{\ast} \left( \left( C, C \right), C \right) = \frac{1}{1 - \gamma} R , \\
x_2 & = Q_A ^{\ast} \left( \left( C, C \right), D \right) = T + \frac{\gamma}{1 - \gamma} S , \\
x_3 & = Q_A ^{\ast} \left( \left( C, D \right), C \right) = Q_A ^{\ast} \left( \left( D, C \right), C \right) = Q_A ^{\ast} \left( \left( D, D \right), C \right) = \frac{1}{1 - \gamma} S , \\
x_4 & = Q_A ^{\ast} \left( \left( C, D \right), D \right) = Q_A ^{\ast} \left( \left( D, C \right), D \right) = Q_A ^{\ast} \left( \left( D, D \right), D \right) = P + \frac{\gamma}{1 - \gamma} S , 
\end{aligned}
\end{equation}
with $x_1 > x_2$ and $x_3 > x_4$. This solution corresponds to the All-C strategy.

\subsubsection{The case $(D, C, D, D)$}

In this case, we obtain
\begin{equation}
\begin{aligned}
x_1 & = Q_A ^{\ast} \left( \left( C, D \right), C \right) = \frac{1}{1 - \gamma^2} (R + \gamma S) , \\
x_2 & = Q_A ^{\ast} \left( \left( C, D \right), D \right) = T + \frac{\gamma}{1 - \gamma^2} (S + \gamma R) , \\
x_3 & = Q_A ^{\ast} \left( \left( C, C \right), C \right) = Q_A ^{\ast} \left( \left( D, C \right), C \right) = Q_A ^{\ast} \left( \left( D, D \right), C \right) = \frac{1}{1 - \gamma^2} (S + \gamma R) , \\
x_4 & = Q_A ^{\ast} \left( \left( C, C \right), D \right) = Q_A ^{\ast} \left( \left( D, C \right), D \right) = Q_A ^{\ast} \left( \left( D, D \right), D \right) = P + \frac{\gamma}{1 - \gamma^2} (S + \gamma R) , 
\end{aligned}
\end{equation}
with $x_1 > x_2$ and $x_3 > x_4$. This solution corresponds to the All-C strategy.

\subsubsection{The case $(D, D, C, D)$}

In this case, we obtain
\begin{equation}
\begin{aligned}
x_1 & = Q_A ^{\ast} \left( \left( D, C \right), C \right) = R + \frac{\gamma}{1 - \gamma} S , \\
x_2 & = Q_A ^{\ast} \left( \left( D, C \right), D \right) = T + \gamma R + \frac{\gamma^2}{1 - \gamma} S , \\
x_3 & = Q_A ^{\ast} \left( \left( C, C \right), C \right) = Q_A ^{\ast} \left( \left( C, D \right), C \right) = Q_A ^{\ast} \left( \left( D, D \right), C \right) = \frac{1}{1 - \gamma} S , \\
x_4 & = Q_A ^{\ast} \left( \left( C, C \right), D \right) = Q_A ^{\ast} \left( \left( C, D \right), D \right) = Q_A ^{\ast} \left( \left( D, D \right), D \right) = P + \frac{\gamma}{1 - \gamma} S , 
\end{aligned}
\end{equation}
with $x_1 > x_2$ and $x_3 > x_4$, when $\gamma < \frac{R - T}{R - S}$, and 
\begin{equation}
\begin{aligned}
x_1 & = Q_A ^{\ast} \left( \left( D, C \right), C \right) = R + \frac{\gamma}{1 - \gamma} S , \\
x_2 & = Q_A ^{\ast} \left( \left( D, C \right), D \right) = \frac{1}{1 - \gamma} T , \\
x_3 & = Q_A ^{\ast} \left( \left( C, C \right), C \right) = Q_A ^{\ast} \left( \left( C, D \right), C \right) = Q_A ^{\ast} \left( \left( D, D \right), C \right) = \frac{1}{1 - \gamma} S , \\
x_4 & = Q_A ^{\ast} \left( \left( C, C \right), D \right) = Q_A ^{\ast} \left( \left( C, D \right), D \right) = Q_A ^{\ast} \left( \left( D, D \right), D \right) = P + \frac{\gamma}{1 - \gamma} S , 
\end{aligned}
\end{equation}
with $x_1 < x_2$ and $x_3 > x_4$, when $\gamma > \frac{R - T}{R - S}$. The former and latter solutions correspond to the All-C and $(C, D, C, C)$ strategies.

\subsubsection{The case $(D, D, D, D)$ (the All-D strategy)}

In this case, we obtain
\begin{equation}
\begin{aligned}
x_3 & = Q_A ^{\ast} \left( \left( C, C \right), C \right) = Q_A ^{\ast} \left( \left( C, D \right), C \right) = Q_A ^{\ast} \left( \left( D, C \right), C \right) = Q_A ^{\ast} \left( \left( D, D \right), C \right) = \frac{1}{1 - \gamma} S , \\
x_4 & = Q_A ^{\ast} \left( \left( C, C \right), D \right) = Q_A ^{\ast} \left( \left( C, D \right), D \right) = Q_A ^{\ast} \left( \left( D, C \right), D \right) = Q_A ^{\ast} \left( \left( D, D \right), D \right) = P + \frac{\gamma}{1 - \gamma} S , \\
\end{aligned}
\end{equation}
with $x_3 > x_4$. This solution corresponds to the All-C strategy.

In light of the above discussion, we can draw a directed graph in which the edges go from each strategy to its optimal response, as shown in Fig.~\ref{BRN}. This graph is referred to as the best-response network in Ref.~\cite{Meylahn2022}. Fig.~\ref{BRN} shows that the graph contains three absorbing states corresponding to the symmetric solution of Sec.~\ref{sym_solutions}, that is, the All-C, WSLS, and FH strategies. Therefore, if the agents update their strategies by taking the best response to the present opponent, as in Ref.~\cite{Meylahn2022}, their strategies converge to one of these three or to the cycle among them.

\begin{figure}
\includegraphics[width = 15.0cm]{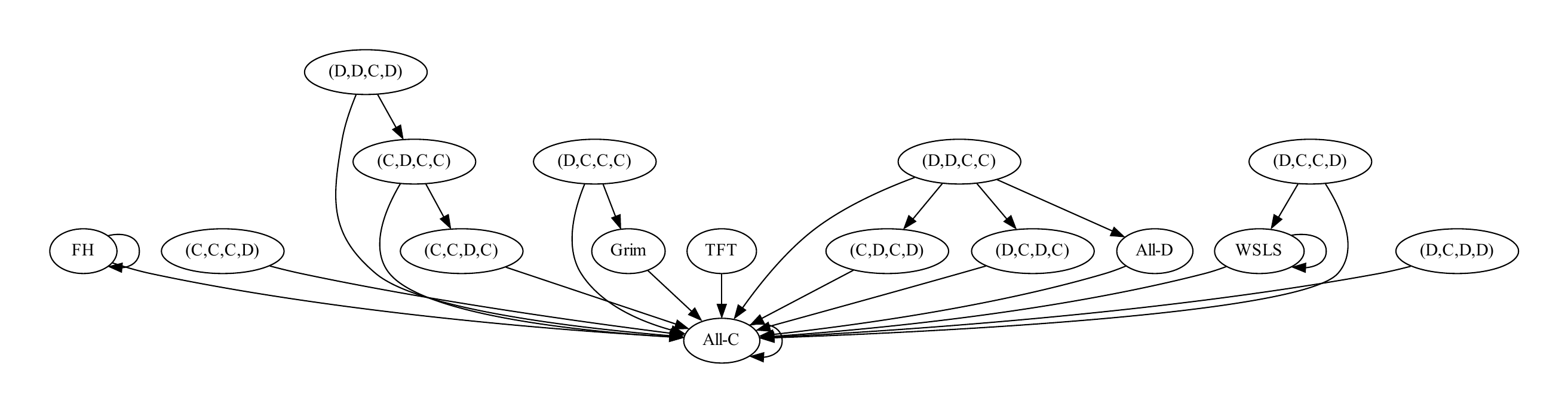}
\caption{The best-response network of the Harmony game when the agents' strategies are deterministic. }
\label{BRN}
\end{figure}

\end{document}